\documentclass[conference]{IEEEtran}

\usepackage[final]{graphicx}
\usepackage[reqno]{amsmath}
\usepackage{amssymb}
\usepackage{url}
\usepackage{color}
\usepackage{times}
\usepackage{wrapfig}
\usepackage{subfig}
\usepackage{amsthm}

\newfont{\bbb}{msbm10 scaled 500}

\newfont{\bb}{msbm10 scaled 1100}

\theoremstyle{plain}
\newtheorem{theorem}{Theorem}

\title{Polar Coding for Fading Channels}

\author{\IEEEauthorblockN{Hongbo~Si, O.~Ozan~Koyluoglu,
and Sriram~Vishwanath}
\IEEEauthorblockA{Laboratory for Informatics, Networks, and Communications\\
Wireless Networking and Communications Group\\
The University of Texas at Austin\\
1 University Station, C0806, Austin, TX 78712\\
Email: \{sihongbo,ozan\}@mail.utexas.edu, sriram@austin.utexas.edu}}

\begin{document}

\maketitle


\begin{abstract}

A polar coding scheme for fading channels is proposed in this paper. More specifically, the focus is Gaussian fading channel with a BPSK modulation technique, where the equivalent channel could be modeled as a binary symmetric channel with varying cross-over probabilities. To deal with variable channel states, a coding scheme of hierarchically utilizing polar codes is proposed. In particular, by observing the polarization of different binary symmetric channels over different fading blocks, each channel use corresponding to a different polarization is modeled as a binary erasure channel such that polar codes could be adopted to encode over blocks. It is shown that the proposed coding scheme, without instantaneous channel state information at the transmitter, achieves the capacity of the corresponding fading binary symmetric
channel, which is constructed from the underlying fading AWGN channel through the modulation scheme.

\end{abstract}


\section{Introduction}

Polar codes are the first family of provably capacity achieving codes for arbitrary symmetric binary-input discrete memoryless channels (B-DMC) with low encoding and decoding complexity \cite{Arikan:Channel08} \cite{Arikan:Error09}. Channel polarization has then been generalized to arbitrary discrete memoryless channels with the same order of construction complexity and error probability behavior \cite{Sasoglu:Polarization09}. Moreover, polar codes are also proved to be optimal for lossy compression with respect to binary symmetric source \cite{Arikan:Source10}\cite{Korada:Source10}, and then further extended to larger source alphabet \cite{Mohammad:Source10}.

Polar codes also contribute significantly to non-discrete input channels. By adopting polar codes as embedded codes at each expanded level, expansion coding scheme \cite{Ozan:Expansion12} achieves the capacity of additive exponential noise channel in high SNR region with low coding complexity. Besides, a polar coding scheme achieving capacity for additive Gaussian noise channel is investigated in \cite{Abbe:AWGN11}, which utilizes the polarization result for multiple access channel \cite{Abbe:MAC12}. It has been shown that the approach of
using a multiple access channel with a large number of binary-input
users has much better complexity attributes than the one of
using a single-user channel with large input cardinality.

In this paper, we investigate the polar coding scheme for binary-input AWGN channel \cite{Arikan:Polar11}\cite{Tse:Polar12} in a fading scheme. By adopting BPSK modulation and demodulation technique, additive Gaussian noise fading channel has been boiled down to a binary symmetric channel (BSC) with finite set of transition probabilities according to the channel quality. The key intuition of the proposed scheme is based on observing the polarization characteristics of different BSCs. By hierarchically using polar codes, where the transmitter encodes over blocks, it can be proved that the designed coding scheme achieves the capacity of converted channel (fading BSC).

The rest of paper is organized as follows. After introducing the preliminary results on polar codes and problem background in Section II and III, respectively, the polar coding scheme for fading channels is stated and illustrated in Section IV. The paper concludes with a discussion section.


\section{Preliminary for Polar Codes}

The construction of polar code is based on the observation of channel polarization. Consider a binary-input discrete memoryless channel $W:\mathcal{X}\to\mathcal{Y}$, where $\mathcal{X}=\{0,1\}$. Define
\begin{equation}F=\left[\begin{array}{cc}
  1 & 0 \\
  1 & 1
\end{array}\right].\nonumber\end{equation}
Let $B_N$ be the bit-reversal operator defined in \cite{Arikan:Channel08}, where $N=2^n$. By applying the transform $G_N=B_NF^{\otimes n}$ ($F^{\otimes n}$ denotes the $n^{\text{th}}$ Kronecker power of $F$) to $u_{1:N}$, consider transmitting the encoded output $x_{1:N}$ through $N$ independent copies of $W$. Then $N$ new binary-input coordinate channels $W_N^{(i)}:\mathcal{X}\to\mathcal{Y}^N\times\mathcal{X}^{i-1}$ are constructed, where for each $i\in\{1,\ldots,N\}$ the transition probability is given by
\begin{equation}
W_N^{(i)}(y_{1:N},u_{1:{i-1}}|u_i)\triangleq \sum_{u_{{i+1}:N}}\frac{1}{2^{N-1}}W^N(y_{1:N}|u_{1:N}G_N).\nonumber
\end{equation}
Then as $N$ tends to infinity, the channels $\{W_N^{(i)}\}$ polarize to either noiseless or pure-noisy, and the fraction of noiseless channels is close to $I(W)$, the symmetric mutual information of channel $W$ \cite{Arikan:Channel08}.

To this end, polar codes can be considered as $G_N$-coset codes with parameter $(N,K,\mathcal{A},u_{\mathcal{A}^c})$, where $u_{\mathcal{A}^c}\in\mathcal{X}^{N-K}$ is frozen vector (can be set to all-zero vector for symmetric channel \cite{Arikan:Channel08}), and the information set $\mathcal{A}$ is chosen as a $K$-element subset of $\{1,\ldots,N\}$ such that the Bhattacharyya parameters satisfies $Z(W_N^{(i)})\leq Z(W_N^{(j)})$ for all $i\in\mathcal{A}$ and $j\in\mathcal{A}^c$.

The decoder in polar coding scheme is successive cancelation (SC) decoder, which gives an estimate $\hat{u}_{1:N}$ of $u_{1:N}$ given knowledge of $\mathcal{A}$, $u_{\mathcal{A}^c}$, and $y_{1:N}$ by computing
\begin{align}
\hat{u}_i \triangleq \left\{
\begin{array}{cl}
  0, & \text{if }i\in\mathcal{A}^c, \\
  d_i(y_{1:N},\hat{u}_{1:{i-1}}), & \text{if }i\in \mathcal{A},
\end{array}
\right.\nonumber
\end{align}
in the order $i$ from $1$ to $N$, where
\begin{align}
d_i(y_{1:N},\hat{u}_{1:{i-1}})\triangleq \left\{
\begin{array}{cl}
  0, & \text{if }\frac{W_N^{(i)}(y_{1:N},\hat{u}_{1:{i-1}}|0)}{W_N^{(i)}(y_{1:N},\hat{u}_{1:{i-1}}|1)}\geq 1, \\
  1,& \text{otherwise.}
\end{array}
\right.\nonumber
\end{align}
It has been proved that by adopting SC decoder, polar codes achieves any rate $R<I(W)$ with decoding error scaling as $O(2^{-N^{\beta}})$, where $\beta<1/2$. Moreover, the encoding and decoding complexity of polar codes are both $O(N\log N)$.

\section{Problem Background}

Fading channels characterize the wireless communication channels, where the channel states are changing over times and only available at the decoders. Fading coefficient typically varies much slower than transmission symbol duration in practice. To this end, a block fading model \cite{Tse:Wireless05} is proposed, whereby the state is assumed to be a constant over coherence time intervals and stationary ergodic across fading blocks.

Consider the AWGN fading channel,
\begin{equation}
Y_{b,i}=H_{b,i}X_{b,i}+Z_{b,i},\quad b=1,\ldots,B,\;i=1,\ldots,N,\label{fun:channel_definition}
\end{equation}
where $Z_{b,i}$ is i.i.d. additive Gaussian noise with variance $E_Z$; $X_{b,i}$ is channel input with power constraint $$\frac{1}{BN}\sum_{b=1}^B\sum_{i=1}^N x_{b,i}^2\leq E_X;$$ $H_{b,i}$ is channel gain random variable; $N$ is blocklength; and $B$ is number of blocks. For this moment, $H_{b,i}$ are assumed to be constant within a block and follow an i.i.d. fading process over blocks. In particular, for the two sates case $\{h_1,h_2\}$ we consider, omitting the indices, the distribution of $H$ is given by $\text{Pr}\{H=h_1\}\triangleq q_1$ and $\text{Pr}\{H=h_2\}\triangleq q_2=1-q_1$.

Using BPSK modulation, any codeword produced by encoder is mapped to signal with element in $\{-\sqrt{E_X},+\sqrt{E_X}\}$. After utilizing a BPSK demodulation at the decoder, the equivalent channel can be formulated as a binary symmetric channel, with transition probability relating to channel states. More specifically, the converted channel is given by
\begin{equation}
\bar{Y}_{b,i}=\bar{X}_{b,i}\oplus \bar{Z}_{b,i},\quad b=1,\ldots,B,\;i=1,\ldots,N,\label{fun:converted_channel}
\end{equation}
where $\bar{X}_{b,i}$ and $\bar{Y}_{b,i}$ are both Bernoulli random variables representing channel input and output correspondingly; $\bar{Z}_{b,i}$ is i.i.d. channel noise, also distributed as Bernoulli random variable, but related to channel state. More precisely, if $H_{b,i}=h_1$, then
\begin{equation}
\text{Pr}\{\bar{Z}_{b,i}=1\}=1-\Phi(h_1\sqrt{\text{SNR}})\triangleq p_1,\label{fun:p_1}
\end{equation}
and if $H_{b,i}=h_2$, then
\begin{equation}
\text{Pr}\{\bar{Z}_{b,i}=1\}=1-\Phi(h_2\sqrt{\text{SNR}})\triangleq p_2,\label{fun:p_2}
\end{equation}
where $\Phi(\cdot)$ is CDF of normal distribution and $\text{SNR}=E_X/E_Z$. In other words, the channel can be modeled as $W_1\triangleq$BSC$(p_1)$ with
probability $q_1$, and as $W_2\triangleq$BSC$(p_2)$ with
probability $q_2$.

The ergodic capacity of the converted channel (fading BSC) is given by \cite{Tse:Wireless05}
\begin{equation}
C_{\text{SI-D}}=q_1[1-H(p_1)]+q_2[1-H(p_2)],\label{fun:fading_capacity}
\end{equation}
where $H(\cdot)$ is the binary entropy function, and SI-D refers to channel state information at the decoder. The capacity achieving input distribution is uniform over $\{0,1\}$. In this paper, we show a polar coding scheme achieving the capacity of converted fading channel with low encoding and decoding complexity, without having instantaneous channel state information at the transmitter (only the statistical knowledge is assumed).

\section{Polar Coding for Fading Channel}
\subsection{Intuition}

In polar coding for a general B-DMC $W$, we have seen the channel can be polarized by transforming a set of independent copies of given channels into a new set of channels whose symmetric capacities tend to 0 or 1 for all but a vanishing fraction of indices. To this end, an information set $\mathcal{A}$ is constructed by picking the indices corresponding to $K$ minimum values of $Z(W_N^{(i)})$, which is equivalent to picking those corresponding to $K$ largest values of $I(W_N^{(i)})$. In this sense, the construction of $\mathcal{A}$ is deterministic. However, as indicated in \cite{Arikan:Channel08}, the indices in $\mathcal{A}$ are not adjacent. For this, we introduce a permutation $\pi:\{1,\ldots,N\}\to\{1,\ldots,N\}$, which reorders all the indices by the value of $I(W_N^{(i)})$ ranging from high to low. Note that the construction of polar codes already implies the fact that for channels of the same type, their permutation mappings are the same.

Another fact about polar codes is that the polarization is uniform \cite{Korada:Thesis09}. Consider polarizing two B-DMCs, for instance BSCs with parameters $p_1$ and $p_2$ respectively, then the information sets, denoted by $\mathcal{A}_1$ and $\mathcal{A}_2$, satisfy
\begin{equation}
\mathcal{A}_1\subseteq \mathcal{A}_2 \;\text{ if }\; p_1\geq p_2. \label{fun:channel_compare}
\end{equation}
In other words, if a particular channel index constructed from the worse channel (BSC with larger transition probability) polarizes to be noiseless, so does that of the better channel (BSC with smaller transition probability).
\begin{figure}[t!]
 \centering
 \includegraphics[width=0.9\columnwidth]{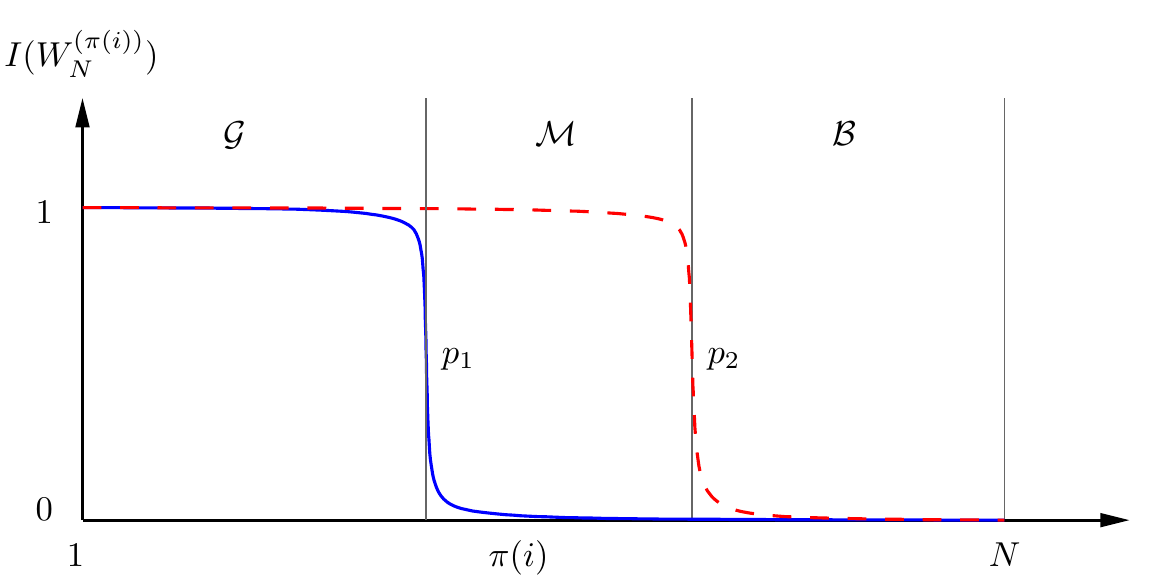}
 \caption{{\bf Illustration of polarizations for two BSCs.} Values of $I(W_N^{(\pi(i))})$, the reordered mutual information, are shown for both polarizations. The blue-solid plot represents the channel with higher transition probability $p_1$, and the red-dashed for $p_2$. Three index categories are denoted by $\mathcal{G}$, $\mathcal{M}$, and $\mathcal{B}$ in order.
}
\label{fig:Polarization}
\end{figure}
\begin{figure*}[t]
 \centering
 \includegraphics[scale=0.7]{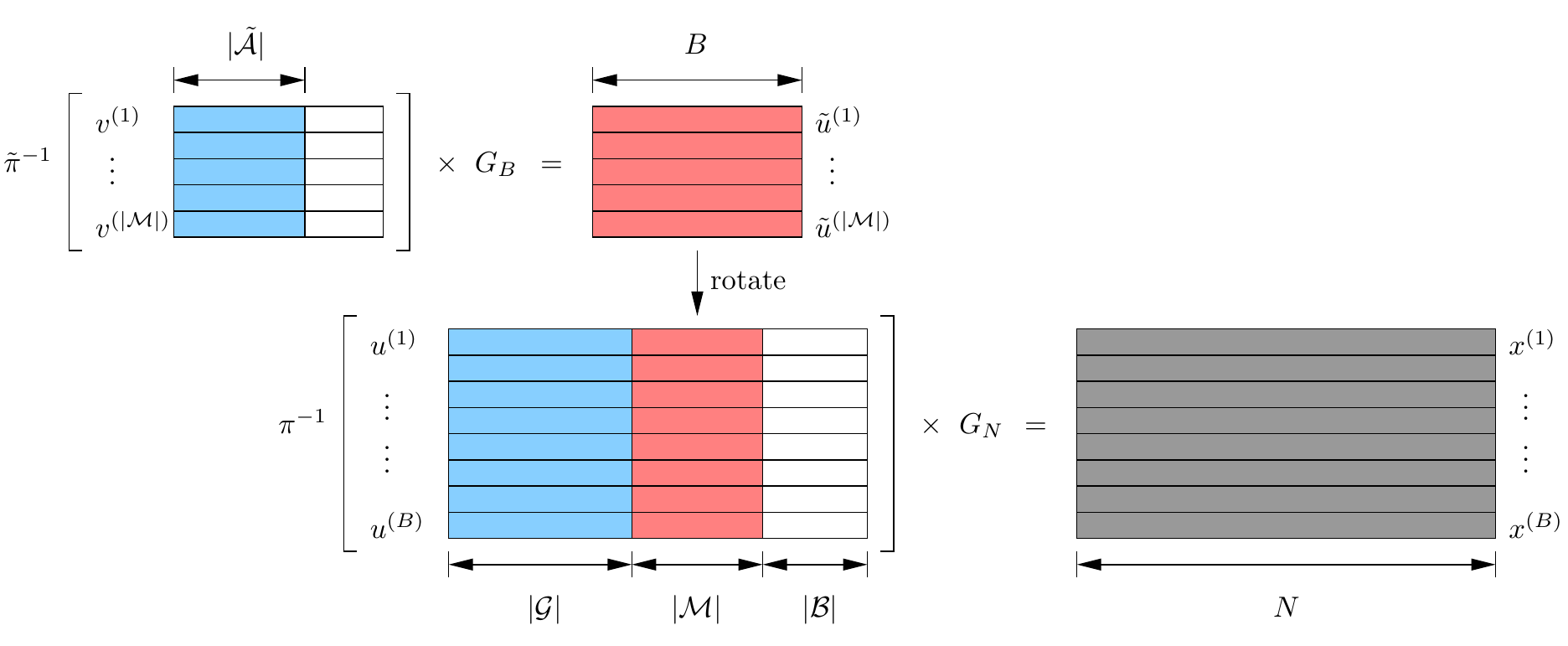}
 \caption{{\bf Illustration of polar encoder for fading channel with two states.} In this numerical example, we assume $N=16$, $B=8$, $|\mathcal{A}|=5$, $|\mathcal{G}|=7$, $|\mathcal{M}|=5$, and $|\mathcal{B}|=4$. Bits in blue are information bits, and ones in white are frozen to zeros. After encoding of Phase 1, the codewords are rotated and embed to the messages of Phase 2 to generate the finalized codeword.
}
\label{fig:Encoder}
\end{figure*}
Based on this observation, when polarizing two $W_1$ and $W_2$ with transition probabilities defined by (\ref{fun:p_1}) and (\ref{fun:p_2}) respectively, the indices after permutation $\pi$ can be divided into three categories (illustrated in Fig.~\ref{fig:Polarization}). Without loss of generality, we assume $p_1\geq p_2$.
\begin{enumerate}
\item $\mathcal{G}$: both channels are good, i.e.
$$I(W_{1,N}^{(\pi(i))})\to 1,\quad I(W_{2,N}^{(\pi(i))})\to 1.$$
\item $\mathcal{M}$: only channel 2 is good, while channel 1 is bad, i.e.
$$I(W_{1,N}^{(\pi(i))})\to 0,\quad I(W_{2,N}^{(\pi(i))})\to 1.$$
\item $\mathcal{B}$: both channels are bad, i.e.
$$I(W_{1,N}^{(\pi(i))})\to 0,\quad I(W_{2,N}^{(\pi(i))})\to 0.$$
\end{enumerate}

Denote the information sets for two channels as $\mathcal{A}_1$ and $\mathcal{A}_2$ correspondingly, then obviously $\mathcal{A}_1=\mathcal{G}$, and $\mathcal{A}_2=\mathcal{G}\cup \mathcal{M}$. Moreover, we have:
\begin{align}
&|\mathcal{G}|=|\mathcal{A}_1|=N[1-H(p_1)-\epsilon],\label{fun:size_G}\\
&|\mathcal{M}|=|\mathcal{A}_2|-|\mathcal{A}_1|=N[H(p_1)-H(p_2)],\label{fun:size_M}\\
&|\mathcal{B}|=N-|\mathcal{A}_2|=N[H(p_2)+\epsilon],\label{fun:size_B}
\end{align}
where $\epsilon$ is a arbitrary small positive number.

For the fading channel, we consider the transmitter has no prior knowledge of channel states before transmitting, hence, coding over channels with indices in $\mathcal{M}$ is challenging. Observe that for those channels, with probability $q_2$ they are nearly noiseless, and with probability $q_1$ they are purely noisy. To this end, each channel can be modeled as a binary erasure channel (BEC) from the viewpoint of blocks, and we denote this channel as $\tilde{W}$. This intuition inspires our design of encoder and decoder for fading channels.

\subsection{Encoder}

The encoding process of polar coding for fading channel has two phases, hierarchically using polar codes to generate $NB$-length codewords, where $N$ is blocklength and $B$ is the number of blocks.
\subsubsection{Phase 1} Consider a set of $B$-length block messages $v^{(k)}$ with $k\in\{1,\ldots,|\mathcal{M}|\}$. For every $v^{(k)}$, construct polar code $\tilde{u}^{(k)}$, which is $G_B$-coset code with parameter $(B,|\tilde{\mathcal{A}}|,\tilde{\mathcal{A}},0)$, where $\tilde{\mathcal{A}}$ is the information set for $\tilde{W}\triangleq \text{BEC}(q_1)$, and we choose
\begin{equation}
|\tilde{\mathcal{A}}|=(1-q_1-\epsilon)B.\label{fun:size_A}
\end{equation}
In other words, we construct a set of polar codes, where each code corresponds to an index in set $\mathcal{M}$, with the same rate $1-q_1-\epsilon$, the same information set $\tilde{\mathcal{A}}$, and the same frozen values 0 as well. Mathematically, if denote the reordering permutation for $\tilde{W}$ as $\tilde{\pi}$, then
    \begin{align}
    &\tilde{\pi}(v^{(k)})=[v^{(k)}_1,\ldots,v^{(k)}_{|\tilde{\mathcal{A}}|},0,\ldots,0],\label{fun:code_phase1}\\
    &\tilde{u}^{(k)}=v^{(k)}G_B.\label{fun:encoding_phase1}
    \end{align}
\subsubsection{Phase 2} Consider another set of $N$-length messages $u^{(b)}$ with $l\in\{1,\ldots, B\}$. For every $u^{(b)}$, construct polar code $x^{(b)}$, which is $G_N$-coset codes with parameter $(N,|\mathcal{G}|,\mathcal{G},u^{(b)}_{\mathcal{G}^c}))$, where $\mathcal{G}$ is BSC information set with size given by (\ref{fun:size_G}). Remarkably, we do not froze all non-information bits to be 0, but embed the blockwise codewords from Phase 1. More precisely, if denote the permutation operator of BSC as $\pi$, then
    \begin{align}
    &\pi(u^{(b)})=[u^{(b)}_1,\ldots,u^{(b)}_{|\mathcal{G}|},\tilde{u}^{(1)}_{b},\ldots,\tilde{u}^{(|\mathcal{M}|)}_{b},0,\ldots,0],\label{fun:code_phase2}\\
    &x^{(b)}=u^{(b)}G_N.\label{fun:encoding_phase2}
    \end{align}
By collecting all $\{x^{(b)}\}_{1:B}$ together, the encoder generates and outputs a codeword with length $NB$. An example to illustrate this encoding process is shown in Fig.~\ref{fig:Encoder}.

\subsection{Decoder}
\begin{figure*}[t]
 \centering
 \includegraphics[scale=0.7]{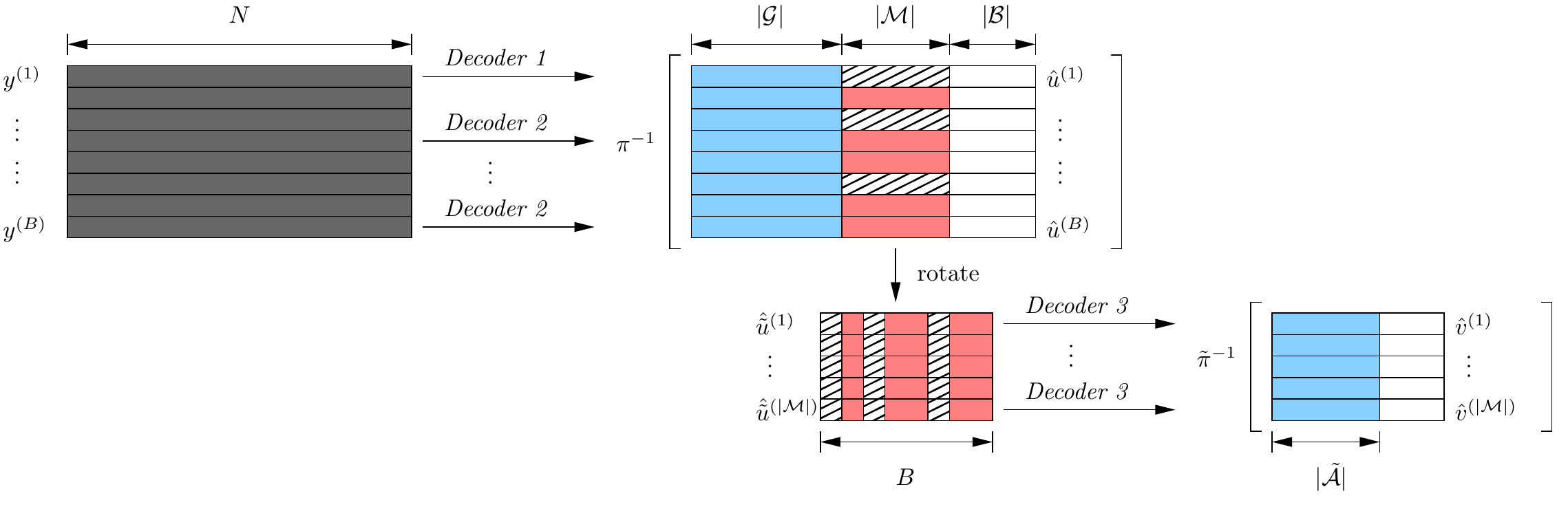}
 \caption{{\bf Illustration of polar decoder for fading channel with two states.} We use the same parameters as encoder. After Phase 1, decoder outputs all estimates $\{\hat{u}^{(l)}\}_{1:B}$ using BSC SC decoder based on channel states, then selected columns of decoded results are rotated and delivered as inputs to Phase 2. In the next phase, the decoder basically uses BEC SC decoder to decode $\hat{v}^{(k)}$ from $\hat{\tilde{u}}^{(k)}$ for every $k\in\{1,\ldots,|\mathcal{M}|\}$. Bits in shade represent for erasures.
}
\label{fig:Decoder}
\end{figure*}

After receiving the sequence $y_{1:{NB}}$ from channel, the decoder's task is trying to make estimates $\{\hat{v}^{(k)}\}_{1:{|\mathcal{M}|}}$ and $\{\hat{u}^{(b)}\}_{1:{B}}$, such that the information bits in both sets of messages match the ones at the transmitter end with high probability. Rewrite channel output $y_{1:{NB}}$ as a $B\times N$ matrix, with row vectors $\{y^{(b)}\}_{1:B}$. As that of the encoding process, the decoding process also has two phases:
\subsubsection{Phase 1} For every $b\in\{1,\ldots,B\}$, decode $\hat{u}^{(b)}$ from $y^{(b)}$ using SC decoder. More precisely,
because at the receiver end, channel state is available, then receiver can adopt the corresponding SC decoder for BSC based on the channel state observed. Remarkably, for index in $\mathcal{M}$, we decode as erasure, denoted as ``e'', for bad channel state. To this end, polar decoder is given by: if the channel state is $h_1$, then use \emph{Decoder 1}, otherwise use \emph{Decoder 2}, where the two decoders are expressed follows: \\
$-$ \emph{Decoder 1:}
    \begin{align}
\hat{u}^{(b)}_i \triangleq \left\{
\begin{array}{ll}
  0, & \text{if }b\in\mathcal{B},\\
  \text{e}, &\text{if }b\in\mathcal{M}, \\
  d_{1,i}(y^{(b)},\hat{u}^{(b)}_{1:i-1}), & \text{if }b\in \mathcal{G},
\end{array}
\right.\nonumber
\end{align}
\indent in the order $i$ from $1$ to $N$, where
\begin{align}
d_{1,i}(y^{(b)},\hat{u}^{(b)}_{1:i-1})\triangleq \left\{
\begin{array}{ll}
  0, & \text{if }\frac{W_{1,N}^{(i)}(y^{(b)},\hat{u}^{(b)}_{1:i-1}|0)}{W_{1,N}^{(i)}(y^{(b)},\hat{u}^{(b)}_{1:i-1}|1)}\geq 1, \\
  1,& \text{otherwise.}
\end{array}
\right.\nonumber
\end{align}
$-$ \emph{Decoder 2:}
    \begin{align}
\hat{u}^{(b)}_i \triangleq \left\{
\begin{array}{ll}
  0, & \text{if }b\in\mathcal{B}, \\
  d_{2,i}(y^{(b)},\hat{u}^{(b)}_{1:i-1}), & \text{if }b\in \mathcal{G}\cup\mathcal{M},
\end{array}
\right.\nonumber
\end{align}
\indent in the order $i$ from $1$ to $N$, where
\begin{align}
d_{2,i}(y^{(b)},\hat{u}^{(b)}_{1:i-1})\triangleq \left\{
\begin{array}{ll}
  0, & \text{if }\frac{W_{2,N}^{(i)}(y^{(b)},\hat{u}^{(b)}_{1:i-1}|0)}{W_{2,N}^{(i)}(y^{(b)},\hat{u}^{(b)}_{1:i-1}|1)}\geq 1, \\
  1,& \text{otherwise.}
\end{array}
\right.\nonumber
\end{align}
After decoding from $y^{(b)}$ block by block, the decoder output a $B\times N$ matrix $\hat{\bold{U}}$ with rows $\{\hat{u}^{(b)}\}_{1:B}$.
\subsubsection{Phase 2} Select columns of $\hat{\bold{U}}$ with indices in $\mathcal{M}$ after permutation $\pi$ to construct
a $B\times |\mathcal{M}|$ matrix $\hat{\tilde{\bold{U}}}$. Consider each column of $\hat{\tilde{\bold{U}}}$, denoted by $\hat{\tilde{u}}^{(k)}$ for $k\in\{1,\ldots,|\mathcal{M}|\}$, as the input to decoder in Phase 2. Then receiver aims to decode $\hat{v}^{(k)}$ from $\hat{\tilde{u}}^{(k)}$ using SC decoder with respect to $\tilde{W}=$BEC$(q_1)$. More formally, the decoder in Phase 2 is expressed as follow:\\
$-$ \emph{Decoder 3:}
    \begin{align}
\hat{v}^{(k)}_j \triangleq \left\{
\begin{array}{ll}
  0, & \text{if }k\in\tilde{\mathcal{A}}^c, \\
  \tilde{d}_j(\hat{\tilde{u}}^{(k)},\hat{v}^{(k)}_{1:j-1}), & \text{if }k\in \tilde{\mathcal{A}},
\end{array}
\right.\nonumber
\end{align}
\indent in the order $j$ from $1$ to $B$, where
\begin{align}
\tilde{d}_j(\hat{\tilde{u}}^{(k)},\hat{v}^{(k)}_{1:j-1})\triangleq \left\{
\begin{array}{cl}
  0, & \text{if }\frac{\tilde{W}_{N}^{(j)}(\hat{\tilde{u}}^{(k)},\hat{v}^{(k)}_{1:j-1}|0)}{ \tilde{W}_{N}^{(j)}(\hat{\tilde{u}}^{(k)},\hat{v}^{(k)}_{1:j-1}|1)}\geq 1, \\
  1,& \text{otherwise.}
\end{array}
\right.\nonumber
\end{align}
After Phase 2, the decoder output a $|\mathcal{M}|\times B$ matrix $\hat{\bold{V}}$ with rows $\{\hat{v}^{(k)}\}_{1:|\mathcal{M}|}$. An example to illustrate the decoding of both phases is shown in Fig.~\ref{fig:Decoder}.

\subsection{Achievable Rate}

We want to show the rate in proposed polar coding scheme achieves the capacity of converted fading channel given by (\ref{fun:fading_capacity}). Intuitively, by using BSC SC decoders corresponding to channel states, the output from Phase 1 successfully recovers all information bits in $\{u^{(b)}\}_{1:B}$. Moreover, for those with indices corresponding to $\mathcal{M}$, the decoder could decode correctly if channel state is $h_2$, and set to erasures otherwise. Thus, the input to decoding Phase 2, vector $\hat{\tilde{u}}^{(k)}$ can be considered as an output of BEC$(q_1)$, hence BEC SC decoder could decode all information bits in $v^{(k)}$ correctly for any $k\in\{1,\ldots,|\mathcal{M}|\}$.

More formally, we have the following theorem.
\begin{theorem}
The proposed polar coding scheme achieves any rate $R<C_{\text{SI-D}}$ with arbitrarily small error probability for sufficiently large $N$ and $B$.
\end{theorem}
\begin{proof}
The proof is straightforward by utilizing error bound from polar coding. In Phase 1 of decoding, the error probability of recovering $u^{(b)}$ correctly for each $b\in\{1,\ldots,B\}$ is given by
\begin{equation}
P_{1,e}^{(b)}=O(2^{-N^{\beta}}),\label{fun:proof_error_bound1}
\end{equation}
where $\beta<1/2$. Similarly, in decoding Phase 2, the error probability of recovering $v^{(k)}$ correctly for each $k\in\{1,\ldots,M\}$ is given by
\begin{equation}
P_{2,e}^{(k)}=O(2^{-B^{\beta}}).\label{fun:proof_error_bound2}
\end{equation}
Hence, by union bound, the total decoding error probability is upper bounded by
\begin{align}
P_{e}&\leq\sum_{b=1}^{B}P_{1,e}^{(b)}+\sum_{k=1}^{|\mathcal{M}|}P_{2,e}^{(k)}\nonumber\\
        &=O(B2^{-N^{\beta}})+O(N2^{-B^{\beta}})\nonumber\\
        &\to 0,\nonumber
\end{align}
when $N$ and $B$ tend to infinity. In particular, we consider $B=o(2^{N^{\beta}})$ and $N=o(2^{B^{\beta}})$.

Moreover, from the analysis, it is evident that all messages bits in $v^{(k)}$ and $u^{(b)}$ are decodable, then the achievable rate for the designed scheme is given by
\begin{align}
R   &=\frac{1}{NB}\Big\{|\mathcal{M}||\tilde{\mathcal{A}}|+B|\mathcal{G}|\Big\}\nonumber\\
    &=\frac{1}{NB}\Big\{N[H(p_1)-H(p_2)]B[1-q_1-\epsilon]\nonumber\\
    &\quad\quad +BN[1-H(p_1)-\epsilon]\Big\}\nonumber\\
    &=q_1[1-H(p_1)]+q_2[1-H(p_2)]-\delta(\epsilon)\nonumber\\
    &=C_{\text{SI-D}}-\delta(\epsilon),\nonumber
\end{align}
where we have used (\ref{fun:size_G}), (\ref{fun:size_M}) and (\ref{fun:size_A}), and
\begin{align}
\delta(\epsilon)\triangleq\epsilon[1+H(p_1)-H(p_2)]\to 0, \text{ as }\epsilon\to 0. \nonumber
\end{align}
Thus, any rate $R<C_{\text{SI-D}}$ is achievable.
\end{proof}
\subsection{Complexity Analysis}

As we have seen, polar coding schemes for both BSC and BEC have relatively low complexity. Since the proposed polar coding scheme for fading channel hierarchically utilizes polar codes, the character of low complexity is consequently inherited. More precisely, $|\mathcal{M}|$ number of $B$-length polar codes as well as $B$ number of $N$-length polar codes are utilized. Hence, the overall complexity of the coding scheme, for both encoding and decoding, is given by
\begin{equation}
|\mathcal{M}|\cdot O(B\log B)+B\cdot O(N\log N)=O(NB\log (NB)).\nonumber
\end{equation}

\section{Discussion}

In this section, we generalize the polar coding scheme to fading channels with arbitrary finite number of states. Assume channel gain $H_{b,i}$ has $S$ states from set $\{h_1,\ldots,h_S\}$, where $S$ is a positive integer. Assume the distribution of $H_{b,i}$ omitting indices is given by  $\text{Pr}\{H=h_s\}\triangleq q_s$, where $s\in\{1,\ldots,S\}$. Then the converted channel using BPSK, defined in (\ref{fun:converted_channel}), is still a BSC, whereas with probability $q_s$, the transition probability is given by
\begin{equation}
\text{Pr}\{\bar{Z}=1\}=1-\Phi(h_s\sqrt{\text{SNR}})\triangleq p_s.
\end{equation}
Denote the converted BSC corresponding to state $h_s$ as $W_s$, then the capacity of converted channel is given by
\begin{equation}
C_{\text{SI-D}}=\sum_{s=1}^S q_s [1-H(p_s)],
\end{equation}
where $1-H(p_s)$ is the capacity of $W_s$.

Observe that when polarizing $S$ BSCs with different transition probabilities, the indices could be divided into $S+1$ sets after permutation $\pi$. More mixture sets $\mathcal{M}_1$, \ldots, $\mathcal{M}_{S-1}$ are defined in this case. Without loss of generality, we assume $p_1\geq p_2\geq\cdots\geq p_S$. Then $|\mathcal{M}_s|=H(p_s)-H(p_{s+1})$, and for index in set $\mathcal{M}_s$,  $W_1,\ldots,W_s$ are polarized to be purely noisy and all others to be noiseless. To this end, we consider a BEC with erasure probability $e_s=\sum_{t=1}^s q_t$ to characterize the polarization result for index in $\mathcal{M}$. (See Fig.~\ref{fig:Polarization_S} for an intuition.)
\begin{figure}[h]
 \centering
 \includegraphics[width=0.9\columnwidth]{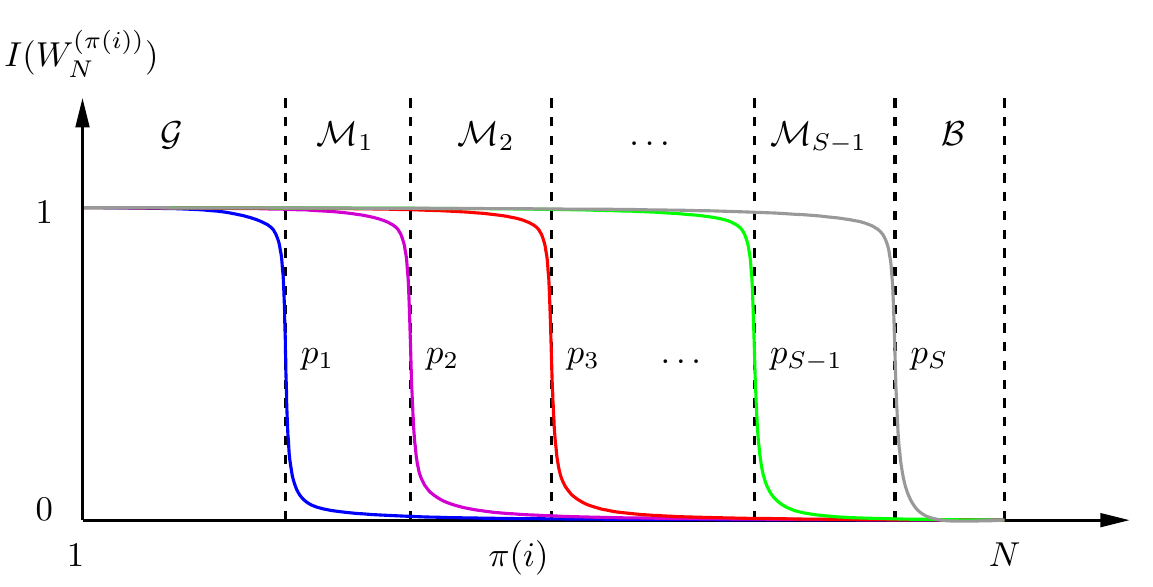}
 \caption{{\bf Illustration of polarizations for $S$ BSCs.} There are $S-1$ mixture sets, denoted as $\mathcal{M}_1,\ldots,\mathcal{M}_{S-1}$.
}
\label{fig:Polarization_S}
\end{figure}

Polar coding scheme designed for this channel is similar. In Phase 1 of encoding, transmitter needs to generate $S-1$ sets of polar codes, where each one is $G_B$-coset codes with parameter $(B,|\tilde{\mathcal{A}}_s|,\tilde{\mathcal{A}}_s,0)$ with respect to BEC$(e_s)$, and all the encoded codewords are embed into messages for Phase 2. At the receiver end, Phase 1 should use one of $S$ SC decoders for BSC to decode $\hat{u}_{(l)}$, based on observation of channel states. Then in Phase 2, $S-1$ BEC SC decoders are working in parallel to recover the information bits. By adopting this polar coding scheme, the achievable rate is given by
\begin{align}
R   &=\frac{1}{NB}\Bigg\{B|\mathcal{G}|+\sum_{s=1}^{S-1}|\mathcal{M}_s||\tilde{\mathcal{A}}_s|\Bigg\}\nonumber\\
    &=[1-H(p_1)-\epsilon]+\sum_{s=1}^{S-1}[H(p_s)-H(p_{s+1})](1-e_s-\epsilon)\nonumber\\
    &=\sum_{s=1}^Sq_s[1-H(p_s)]-\delta'(\epsilon),\nonumber
\end{align}
where $\delta'(\epsilon)=\epsilon[1+H(p_1)-H(p_S)]$. Thus, the proposed polar coding scheme achieves the capacity of channel, and the encoding and decoding complexities are both given by
\begin{equation}
\sum_{s=1}^{S-1}|\mathcal{M}_s|\cdot O(B\log B)+B\cdot O(N\log N)=O(NB\log (NB)),\nonumber
\end{equation}
which is irrelative to the value of $S$, as $\sum\limits_{s=1}^{S}|\mathcal{M}_s|\leq N$. 
%
%


\end{document}